\documentclass{ecai} 

\usepackage{latexsym}
\usepackage{amssymb}
\usepackage{amsmath}
\usepackage{amsthm}
\usepackage{booktabs}
\usepackage{enumitem}
\usepackage{graphicx}
\usepackage{color}
\usepackage[small]{caption}
\usepackage{subcaption}

\newtheorem{theorem}{Theorem}

\newtheorem{definition}{Definition}

\newcommand{\BibTeX}{B\kern-.05em{\sc i\kern-.025em b}\kern-.08em\TeX}

\begin{document}


\begin{frontmatter}


\paperid{123} 


\title{Degree Matrix Comparison for Graph Alignment}


\author[A]{\fnms{Ashley}~\snm{Wang}\thanks{Corresponding Author. Email: ashley.m.wang.25@dartmouth.edu.}
}
\author[B]{\fnms{Peter}~\snm{Chin}}

\address[A]{Department of Mathematics, Dartmouth College}
\address[B]{Thayer School of Engineering}


\begin{abstract}
The graph alignment problem, which considers the optimal node correspondence across networks, has recently gained significant attention due to its wide applications. There are graph alignment methods suited for various network types, but we focus on the unsupervised geometric alignment algorithms. We propose Degree Matrix Comparison (DMC), a very simple degree-based method that has shown to be effective for heterogeneous networks. Through extensive experiments and mathematical proofs, we demonstrate the potential of this method. Remarkably, DMC achieves up to 99\% correct node alignment for 90\%-overlap networks and 100\% accuracy for isomorphic graphs. Additionally, we propose a reduced Greedy DMC with lower time complexity and Weighted DMC that has demonstrated potential for aligning weighted graphs. Positive results from applying Greedy DMC and the Weighted DMC furthermore speaks to the validity and potential of the DMC. The sequence of DMC methods could significantly impact graph alignment, offering reliable solutions for the task.
\end{abstract}

\end{frontmatter}


\section{Introduction}
Graph alignment is a critical task with applications across various domains, like social networks, biology, and cybersecurity. In social networks, graph alignment can identify the same user across multiple platforms, providing insights into user behavior and preferences. In biology, it helps in comparing molecular structures and understanding functional similarities between biological networks. In cybersecurity, aligning graphs can detect coordinated attacks and identify malicious entities by correlating data from different sources. Graph alignment is also highly related to the subgraph isomorphism problem in mathematics and computer science, which is an NP-complete problem. Improving the efficiency and simplifying implementation steps of graph alignment algorithms are important in artificial intelligence because they help us learn patterns on networks. 

There are supervised graph alignment methods that work well already \citep{LiCaoHou:semi-supervised-example}. This is not surprising as supervised and semi-supervised methods exploit node attributes, such as semantic features of users \citep{Staab2009,Yang2021}. However, they break down when users disguise their identities and node attributes are no longer reflective of local properties \citep{Allison2024}. Moreover, they require some knowledge of the node labeling on a network prior to the alignment. There are also unsupervised methods that rely on the power of node attributes, like WAlign \citep{Gao2021}. We, however, focus on the unsupervised graph alignment problem for unattributed plain graphs, mainly leveraging geometric properties. 

Our proposed method, Degree Matrix Comparison (DMC), is a method that fits the task. We list some previous unsupervised methods that exploit graph geometries: REGAL, which is a representation learning-based method \citep{Heimann2018}; FINAL, which mainly utilizes graph topology but enhances it with node attributes \citep{Zhang2016}; Klau, which utilizes a Lagrangian relaxation approach in an optimization problem \citep{Klau2009}; and IsoRank, which constructs an eigenvalue problem for every pair of input graphs and then uses k-partite matching techniques \citep{Singh2008}. For a more comprehensive view of past unsupervised alignment methods, see \citet{Skitsas2023}. 

DMC was constructed with the following two core ideas in mind. First of all, degree is the most accurate and direct description of local structure. Moreover, in representing geometric structure, it makes sense to include both local and global information, like in \citet{Kuchaiev2010,Milano2016}, especially in compact matrix form (like in adjacency matrices and Laplacian matrices). This is not just because ``the more the better'', rather it is a consequence of the Friendship Paradox. More specifically, we propose foregoing the need to know exact connections between nodes, like in an adjacency matrix, and record accurately two-layer exact information--degree of a node and its neighbors' degrees in a matrix--in a clever way, for later use of the Hungarian algorithm. Since we are aligning unattributed graphs, a method that is based on information about links between labeled nodes is not rigorous, because the assignment of node IDs could be random on the pair of graphs to be aligned. Although the Hungarian algorithm is widely known and has been integrated into other graph alignment methods \citep{Bougleux2017}, our approach merits further attention and research due to its exceptional applicability to real heterogeneous networks--rather than restricted ideal graphs like k-partite graphs or complete graphs--with very simple implementation steps.

\section{Problem Set Up}
\begin{definition}[Graph] A \emph{graph} is a pair $G = (V, E)$, where elements $v \in V$ are called vertices (or nodes), and $E$ is a set consisting of unordered (for our purposes only) pairs $\{v_{1}, v_{2}\}$ for some but not necessarily all of $v_{1}, v_{2} \in V$. 

\end{definition}

Geometrically, $G$ could represent an object with vertices and edges that connect some of the vertices. The presence of $\{v_{1}, v_{2}\} \in E$ indicates there is an edge between the vertices $v_{1}, v_{2}$. We talk about graphs in a manner that assumes a geometric interpretation. But there are other ways to interpret what we call a ``graph''; a most direct example is the adjacency matrix, which is a binary matrix representation of graphs. 

\begin{definition}[Graph Alignment]
If we have $G_{1} = (V_{1}, E_{1})$ and $G_{2} = (V_{2}, E_{2})$ for graph alignment, we try to find the optimal mapping $f: V_{1} \rightarrow V_{2}$ to maximize the number of $v_{1} \in C \subseteq V_{1}$ such that $v_{1} = v_{2} = f(v_{1}) \in V_{2}$. This process is called \emph{Graph Alignment}.
\end{definition}
To perform graph alignment, we generate a pair of graphs with common structure. We introduce two graph sampling methods that sample size-wise manageable smaller graphs from larger complex graphs. 

\subsection{Random Walk Graph Sampling}
One of our methods utilizes random walk on graphs to extract subgraphs \citep{Allison2024,Leskovec2006}.

\begin{definition}[Random Walk (on graphs)] Suppose we start random walk at $v \in V$. We first find the set of neighbors of $v$, which are vertices connected to $v$ by at least one edge, that we denote as $S_{nei}$. Then we randomly pick $v_{n} \in S_{nei}$, and move from $v$ to $v_{n}$. Then we do the same thing for $v_{n}$, and all following nodes iteratively, until called to stop. This process is called random walk.
\end{definition}

We take subgraph $G_{s} = (V_{s}, E_{s})$, with $n = |V_{s}|$ distinct nodes, of a real world network $G_{r} = (V_{r}, E_{r})$ through random walk. This is achieved through the following procedures: we pick random $v_{0} \in V_{r}$, and initiate random walk on $G_{r}$ until we either get one connected component (a graph that cannot be represented as the disjoint union of two graphs) with $n$ distinct nodes, or exhaust a connected component in $G_{r}$ with insufficient nodes, whichever comes first. If the latter is true, we randomly pick a new (different from all previous nodes that we walked on) node $v_{n} \in V_{r}$ and continue random walk until we obtain a sampled graph with the desired $n$ distinct nodes and as few connected components as possible. $V_{s}$ is the set that contains the $n$ distinct nodes from the random walk process just described. $E_{s} \subseteq E_{r}$ inherits edges between sampled nodes in $V_{s}$ from $G_{r}$. 

To construct $G_{1} = (V_{1}, E_{1}) \subseteq G_{s}$ and $G_{2} = (V_{2}, E_{2}) \subseteq G_{s}$ for graph alignment, we perform random walk on $G_{s}$ to first obtain a set of nodes that serves as the common set of nodes $C = V_{1} \bigcap V_{2} \subset G_{s}$. The number of vertices we want to stop at, for the random walk, is $|C| = np$, where $n$ is like before the number of distinct nodes in $G_{s}$ and $p$ is what we call the ``overlap percentage" (with respect to $G_{s}$). We then take arbitrary $H_{1}, H_{2} \subseteq (V_{s}\setminus C)$ such that $|H_{1}| = |H_{2}|$ and $V_{s}\setminus C = H_{1} \sqcup H_{2}$ (disjoint union). Let $V_{1} = C \sqcup H_{1}$ and $V_{2} = C \sqcup H_{2}$, then $|V_{1}| = |V_{2}|$. $E_{1} \subseteq E_{s}$ is the inherited set of edges between nodes in $V_{1}$ from $E_{s}$, likewise for $E_{2}$. 

Figure~\ref{fig:randomW} is a demonstration for what we are trying to take off the entire graph. There are three green circles containing different regions of the graph; the two larger ones contain the same number of nodes, and the smallest circle contains the common region between the two possible sampled graphs through random walk. This is not the only sampling possible.

\begin{figure}
    \centering
\includegraphics[width=.3\textwidth]{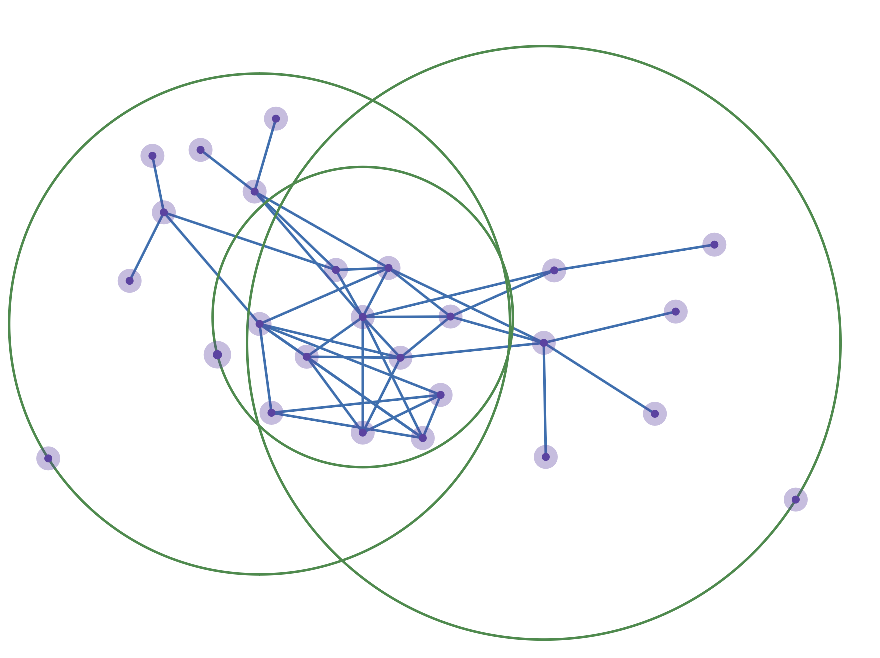}
    \caption{Graph Sampling with Random Walk}
    \label{fig:randomW}
\end{figure}

\subsection{Edge Deletion Graph Sampling}
The second graph sampling method is from \citet{Heimann2018}. Suppose $G_{s}$ is the same graph as in Section 2.1, then let $G_{1} = G_{s} = (V_{1}, E_{1})$, and we delete edges of $G_{1}$ with probability $p_{d}$ to obtain another graph, which is $G_{2}$. If $G_{2} = (V_{2}, E_{2})$, then $V_{2} = V_{1}$, and $E_{2} \subseteq E_{1}$.  
\begin{figure}
    \centering    \includegraphics[width=.35\textwidth]{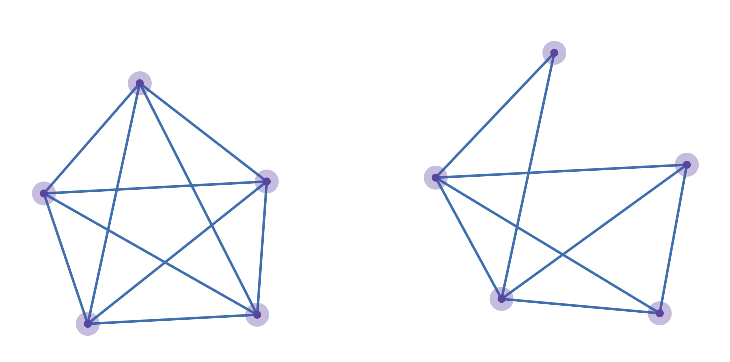}
    \caption{Graph Sampling with Edge Deletion}
    \label{fig:edgeD}
    \vspace{0.5cm}
\end{figure}

Figure~\ref{fig:edgeD} is a simple demonstration of the edge deletion graph sampling method, where two edges from the left graph were randomly deleted to obtain the right graph.

\section{Proposed Graph Alignment Methods}
\subsection{Degree Matrix Comparison}
Our main proposed method requires creating two \emph{degree matrices} and comparing them. We denote $\text{deg}(v)$ as the degree of a node $v$, where degree is simply the number of edges connected to a vertex, or more formally, the number of unordered pairs in the edge set $E$ that contains the vertex.

\begin{definition}[Degree Matrices]
Consider a pair of graphs $G_{1} = (V_{1}, E_{1})$ and $G_{2} = (V_{2}, E_{2})$. By our problem set up, $N = |V_{1}| = |V_{2}|$. Let $V_{1} = \{v_{i}\}_{i=1}^{N}$ and $V_{2} = \{w_{j}\}_{j=1}^{N}$. Let $m_{1} = \text{max}(\{\text{deg}(v_{i})\}_{i=1}^{N})$ and $m_{2} = \text{max}(\{\text{deg}(w_{j})\}_{j=1}^{N})$, and $m = \text{max}(m_{1}, m_{2})$. Then the \emph{degree matrix} $M_{1}$ for $G_{1}$ is a matrix of dimension $(N, m)$, where each row left-aligns, for a unique $v_{1} \in V_{1}$, its neighbors' degrees in ascending order, with zero-padding to the right. We define $M_{2}$ similarly.
\end{definition}
The order of rows in $M_{1}$ and $M_{2}$ is arbitrary. After obtaining $M_{1}$ and $M_{2}$, we use the Hungarian algorithm to minimize cost between rows of the matrices, producing a map between rows of $M_{1}$ and $M_{2}$. We migrate this assignment directly to the assignment $f: V_{1} \rightarrow V_{2}$ between nodes. We call this entire process the \emph{Degree Matrix Comparison (DMC)}. 

We provide examples here to demonstrate the concept of a degree matrix and the DMC process. For the degree matrix formulation, let us examine the graph (assumed to be a local picture taken from a much larger graph) in Figure~\ref{fig:DMCExample} with circled node as the origin. Then the boxed nodes are the first neighbors of the origin. Starting from the first neighbor in the upper left corner and going counter clock-wise, the degrees of the five first neighbors are $3, 4, 3, 5, 3$. Suppose (arbitrarily) $m$ in this case is $10$, then the row vector representation for the circled node would be $[3, 3, 3, 4, 5, 0, 0, 0, 0, 0]$ in the degree matrix for the larger graph. We do this for every node in the larger graph which this smaller piece is taken from to build a degree matrix. 
\begin{figure}
    \centering    \includegraphics[width=0.15\textwidth]{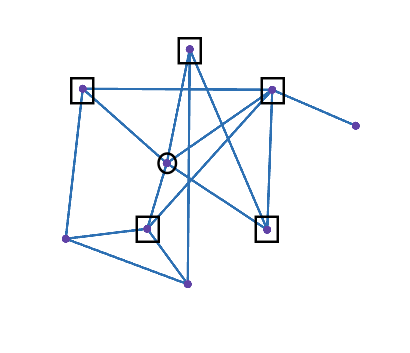}
    \caption{Example Graph to Demonstrate Degree Matrix Formulation}
    \label{fig:DMCExample}
\end{figure} 
\begin{figure}
    \centering
\includegraphics[width=0.2\textwidth]{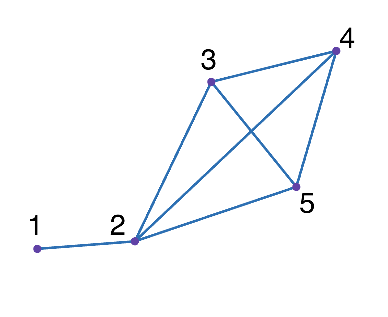}
        \caption{$F_{1}$}
        \label{fig:subfig1}
        \vspace{0.3cm}
\end{figure}
\begin{figure}
    \centering
        \includegraphics[width=0.2\textwidth]{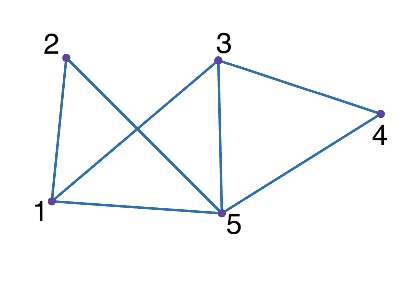}
        \caption{$F_{2}$}
        \label{fig:subfig2}
        \vspace{0.6cm}
\end{figure}

Now we demonstrate the DMC process. Consider the two graphs in Figures~\ref{fig:subfig1} and~\ref{fig:subfig2} as a pair that we are to align. Let the graph in Figure~\ref{fig:subfig1} be $F_{1} = (V_{F_{1}}, E_{F_{1}})$ and the graph in Figure~\ref{fig:subfig2} be $F_{2} = (V_{F_{2}}, E_{F_{2}})$. The node set $V_{l}$ of the largest overlapping subgraph consists of all five nodes, or: $V_{l} = V_{F_{1}} = V_{F_{2}}$. However, the graphs are not isomorphic, so the edge set $E_{l}$ of the overlapping subgraph is a strict subset of $E_{F_{1}}, E_{F_{2}}$, which in mathematical notation is: $E_{l} \subsetneq E_{F_{1}}, E_{F_{2}}$. Let the edge between node 3 and node 4 in $F_{2}$ be $w$. Then the overlapping subgraph takes the shape of $F_{2}\backslash \{w\}$. The degree matrices for $F_{1}$ and $F_{2}$ are

\[
\begin{array}{c@{\quad\text{and}\quad}c}
    \begin{bmatrix}
        4 & 0 & 0 & 0 \\
1 & 3 & 3 & 3 \\
3 & 3 & 4 & 0 \\
3 & 3 & 4 & 0 \\
3 & 3 & 4 & 0
    \end{bmatrix}
    &
    \begin{bmatrix}
        2 & 3 & 4 & 0 \\
3 & 4 & 0 & 0 \\
2 & 3 & 4 & 0 \\
3 & 4 & 0 & 0 \\
2 & 2 & 3 & 3 \\
    \end{bmatrix}
\end{array}
\]
Then we use the Hungarian algorithm to minimize the cost of row alignment between the two matrices above. The algorithm aligns nodes $1,2,3,4,5$ of $F_{1}$ to nodes $2,5,1,3,4$ of $F_{2}$ (one could check understanding of the concepts with this result). This is not an optimal result, but it is expected since the graphs are not heterogeneous (the relationship between performance of DMC and heterogeneity of graph will be discussed in depth in Section 5). This example is purely for demonstrating DMC.

In practice (implementing in code), the Hungarian algorithm can be replaced by a computationally more efficient  \textcolor{blue}{\texttt{linear\_sum\_assignment}} function from \textcolor{blue}{\texttt{scipy.optimize}} in Python which yields the same assignment \citep{Crouse2016}. However, in terms of comprehension, the Hungarian algorithm is more straightforward. The complexity of DMC is at most $O(N^{3})$ for a pair of graphs that each has $N$ nodes \citep{Edmonds1972,Tomizawa1971}. When $m << N$, the complexity is $O(N^{2}\cdot m) << O(N^{3})$.

\subsection{Variation 1: Greedy DMC}
The complexity of DMC can be reduced through replacing  the Hungarian algorithm with a combination of the auction algorithm and Hungarian algorithm, but accuracy is compromised \citep{Bertsekas1981}. The auction algorithm, with $O(N^{2})$ complexity and assignment threshold $\epsilon$, serves as a preparation step to reduce reliance on the Hungarian algorithm. We call this reduced version the \emph{Greedy DMC}. 

We bring to attention the case where $\epsilon = 0$. While this is a special case of the Greedy DMC algorithms with different $\epsilon$ thresholds, it is not ``Greedy'' in any sense. In fact, this is an algorithm that first finds \emph{exact} matches between rows in $M_{1}$ and $M_{2}$, then performs the Hungarian algorithm on remaining rows. As can be seen in Section 5, Greedy DMC's $\epsilon=0,10$ cases are very effective algorithms like DMC, and are worth studying in the future.

\subsection{Variation 2: Weighted DMC}
While the previous methods are meant for unweighted plain graphs, the Weighted DMC is designed for weighted graphs. Graphs can no longer be fully represented by just an adjacency matrix, which counts number of edges but makes no note of varying weights. The idea is simple: replace each entry in the degree matrix with its original value multiplied by the weight of the edge connecting the row’s node and the node the entry represents. 

Weighted DMC performed well on multiple datasets, suggesting it deserves further research--much like Greedy DMC. In particular, its strong performance on protein-protein interaction (PPI) networks (often used to test graph alignment) indicates that DMC’s success is not coincidental.

\section{Theoretical Motivations}
\subsection{Invariance}
One theoretical motivation for the method is that degree matrices remain \emph{invariant} in response to arbitrary node labeling. By invariant, we mean that degree matrices are in the same equivalence class defined by the following equivalence relation. 
\begin{theorem}[Equivalence Relation on Degree Matrices.]
We define a binary relation between two degree matrices: degree matrices $A,B$ are said to be related (i.e. $A \sim B$) if we can transform $A$ to become $B$ only through row swap operations. Moreover, we claim that this is an equivalence relation.
\end{theorem}
\begin{proof}
A binary relation is an equivalence relation if and only if the relation is reflexive, symmetric, and transitive. For any matrix $A$, $A \sim A$ is true since we can obtain $A$ from performing no row swap operations on $A$, hence the relation is reflexive. 

$A \sim B$ implies there is a series of row swaps that can transform $A$ to $B$. Therefore, we can also obtain $A$ from $B$ through reversing the row swap process and so $B\sim A$. Hence, the relation is symmetric.

$A \sim B$ and $B \sim C$ implies we can transform $A$ into $B$ through row swaps, and also transform $B$ into $C$ through row swaps. This implies we can obtain $C$ from $A$ through first performing row swaps that transform $A$ into $B$ and then applying more row swaps that transform $B$ into $C$. Hence we can obtain $C$ from $A$ through row swaps and $A\sim C$, making the binary relation transitive.
\end{proof}

This is crucial for graph alignment for unattributed graphs (or on attributed graphs but not using attributes). 
\subsection{Performance Analysis}
We provide an analysis of performance of the DMC, which shows that the DMC is effective in differentiating high- and low-degree nodes. In a heterogeneous graph similar to one generated by the Barabási-Albert model (introduced in Section 5), a significant amount of nodes are either a concentrated hub center or close to a leaf due to the formation of hubs. Therefore, we can assume the row vector in a degree matrix for a high-degree node (a hub) looks like
\begin{equation}
[1, 1, 1, 1, 1, 1, 1, 2, 3, 0, 0]
\end{equation}
just as an example. The main properties of this row vector are that entries are small numbers (low-degree nodes) and there are many non-zero entries, due to node's large degree. For a leaf, the row vector will look like
\begin{equation}
[100, 0, 0, 0, 0, 0, 0],
\end{equation}
for instance. In these low-degree node row vectors, we observe that most entries are zero, and the few non-zero entries tend to be large. To perform a cleaner analysis, let us directly examine the following row vectors:
\begin{equation}
\mathbf v_{1} = [k, 0, 0, ..., 0]
\end{equation}
for a leaf connected to a hub with $\text{deg}(v) = k$, and 
\begin{equation}
\mathbf v_{2}=[1, 1, 1, ..., 1]
\end{equation}
for the hub, assumed to be connecting to $k$ leaves for simplicity. Both of these row vectors have a length of $k$, since extra 0's hanging at the end does not affect cost of alignment. Moreover, suppose there is another high-degree node row vector with the same length which takes the form
\begin{equation}
\mathbf v_{3} = [1+\epsilon, 1+\epsilon, ..., 1+\epsilon].
\end{equation}
One might wonder whether the accumulation of small differences between each entry in the row vectors lead to a large enough difference that confuse the alignment of a high-degree node with another high-degree node versus with a low-degree node. 

We assume $k \rightarrow \infty$, and $\epsilon \rightarrow 0$. Then the cost of aligning $\mathbf v_{1}$ and $\mathbf v_{2}$ is $(k-1)^{2} + (k-1) \cdot (0-1)^{2} = k(k-1)\approx k^{2}$. This is the cost of aligning a low-degree row vector with a high-degree row vector. We also compute the cost of aligning $\mathbf v_{2}$ with $\mathbf v_{3}$, which is $k\cdot (1+\epsilon - 1)^{2} = k\epsilon^{2}\approx 0$. Even if we relax our assumption so that $\epsilon \rightarrow 0< C < \sqrt{k}$ rather than having $\epsilon \rightarrow 0$, the squared effect will make the alignment of high to high more reasonable.

\subsection{More General Theoretical Motivations}
A degree matrix ($N\times m$) encodes information more compactly than an adjacency matrix ($N\times N$), since we contain both degree distribution and each node's local geometric details in a smaller matrix. Additionally, combining local (degrees) and global (assignment) perspectives is important because the of the Friendship Paradox. In general terms, it is the observation that friends of a person typically have more friends than that person. The Friendship Paradox shows how local observations can be completely shifted when we examine the same neighborhood globally \citep{Alipourfard2020,Lerman2016}. This property of graphs has made it possible for multiple definitions of ``average degree'' to exist on different scales, and hence led to many ways of formulating the problem to be addressed in the paradox. 
Intuitively, our method works well for heterogeneous graphs for the following reason. Suppose there is a graph that has nodes with degrees that are almost all unique. Note here that the phrase ``almost all unique'' is used to carefully address the possible scenario as follows: if we allow at most one edge between every pair of vertices on a graph $G_{example}$ with $n_{example}$ vertices and require at least degree one for all nodes, then the maximum degree is at most $n_{example} - 1$ for every node, and in this case it is impossible to create a bijection between possible degrees and the nodes. With a heterogeneous graph with most degrees different, most nodes have row vectors that are relatively unique and distinguishable, compared to, for example, a complete graph.

\section{Experimental Results}
We first list here the sources of downloaded datasets. From COSNET \citep{Zhang2015}, we downloaded the Flickr, Last.fm, and Myspace networks; from SNAP \citep{SNAP}, we downloaded the Facebook and YouTube networks, along with the CA and PA roadmaps, and a multi-layer tissue PPI network (unweighted); from STRING \citep{STRING}, we downloaded five weighted PPI networks for different species that are typical subjects in biological studies; all other networks are taken from Network Repository \citep{NetworkRepository}. Before we dive into results for specific graphs, we note here that all isomorphic graph alignment yielded 100\% correct results.
\subsection{Unweighted Graphs: DMC and Greedy DMC}
\subsubsection{Real World Networks}
In Table~\ref{tab:baselines}, we display statistics that compare DMC and Greedy DMC with baselines after applying all methods to the unweighted multi-layer PPI network \citep{Gao2021,Heimann2018,Singh2008,Zhang2016}. Using the edge deletion sampling method, we set deletion probability $p_{d} = 0.01$. Regarding the different graph alignment methods, we note that the Greedy DMC uses a threshold of $\epsilon = 100$. ``Two-Step'' is, as its name suggests, a two-step correction method that uses degrees for alignment, and refines that first-step alignment through average degrees of neighbors. ``Euc Dist.'' aligns two graphs by minimizing Euclidean distances between embeddings. We can observe through direct comparison that DMC has highest accuracy, but larger complexity compared to REGAL and FINAL \citep{Heimann2018,Zhang2016}. The Greedy DMC strikes a similar balance compared to FINAL: while it is more accurate, it could also be higher in time complexity. 
\begin{table}
    \centering
    \caption{Comparison with baselines on PPI network.}
    \begin{tabular}{lrr}
        \toprule
        Method  & Score (\%) & Complexity \\
        \midrule
        \textbf{DMC}     & 96.09          &     $\sim O(n^{2})-O(n^{3})$    \\
        \textbf{Greedy DMC}  & 50.46          & $\sim O(n^{2})-O(n^{3})$        \\
        REGAL  &    $80<p<90$       &    $\sim O(n)-O(n^{2})$    \\
        FINAL  & $35<p<45$          & $\sim O(n^{2})$        \\
        Klau  & $20<p<30$         & $\sim O(n^{5})$        \\
        IsoRank  & $5<p<15$          & $\sim O(n^{4})$        \\
        Two-Step  & 4.42          & $\sim O(n^{2})-O(n^{3})$        \\
        Euc Dist.  & 0.15 (near zero)          & $\sim O(n^{2})-O(n^{3})$        \\
        \bottomrule
    \end{tabular}
    
    \label{tab:baselines}
\end{table}

After showing the potential of DMC and Greedy DMC by comparing to baselines on a classic network, we turn to demonstrating that DMC works better when networks are more heterogeneous. We first draw an observed correlation between variance in degree distribution and performance of the DMC; the correlation is not strict, but implies from at least one perspective that the DMC works well for heterogeneous networks in general. Table~\ref{tab:MandV} provides a summary of statistics for social networks and roadmaps, and Figures~\ref{fig:subfigSocial} and~\ref{fig:subfigRoad} provides results applying DMC to the networks. The best performing Facebook and YouTube have large variances in their degree distributions. Comparing Flickr, Last.fm, and Myspace, Flickr and Last.fm are more suitable for DMC, which is most likely due to their much larger variance in degree distribution. However, the correlation between variance and performance can be false at times, as can be seen from the California and Pennsylvania roadmaps. The two roadmaps have close degree distribution means and variances, with the Pennsylvania one having lower degree variance, but it outperforms the California roadmap significantly. 
\begin{table}
  \centering
  \caption{Means and variances of degree distributions of $G_{s}$.}
  \begin{tabular}{lrrr}
    \toprule
    Networks & Mean & Variance & Total Nodes ($N$)\\ \midrule
    Flickr & 52.36 & 5330.29 & 5000\\  
    Last.fm & 26.03 & 2164.02 & 5000\\
    Myspace & 7.59 & 113.32 & 5000\\
    Facebook & 85.06 & 4040.27 & 700\\ 
    YouTube & 14.52 & 960.27 & 5000\\  
    CA Roadmap & 2.55 & 0.96 & 1000\\  
    PA Roadmap &  2.88 & 0.78 & 1000\\ \bottomrule
  \end{tabular}
  
  \label{tab:MandV}
\end{table}
\begin{figure}
    \centering
    \includegraphics[width=0.35\textwidth]{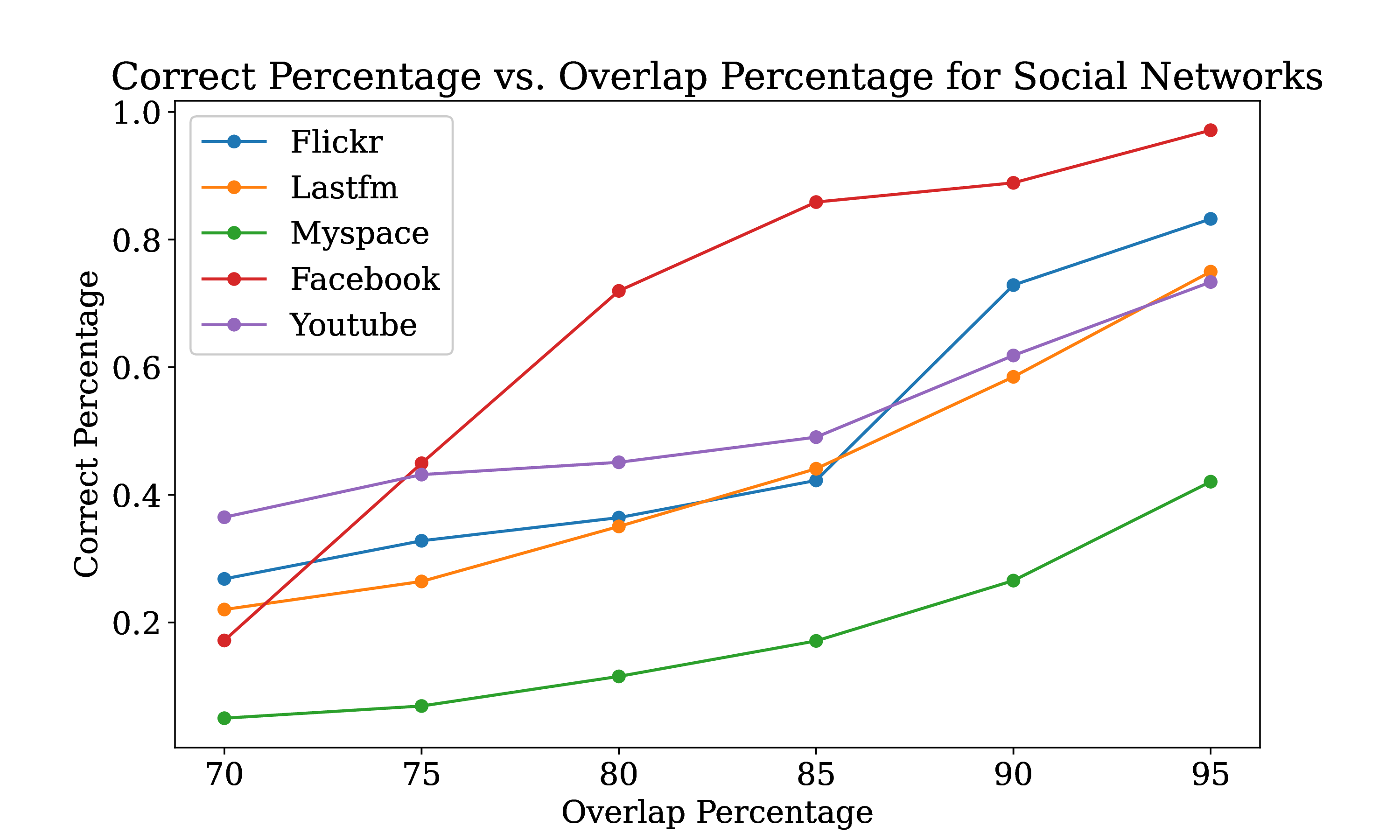}
        \caption{Social Networks}
        \label{fig:subfigSocial}
        \vspace{0.6cm}
\end{figure}
\begin{figure}
    \centering
        \includegraphics[width=0.35\textwidth]{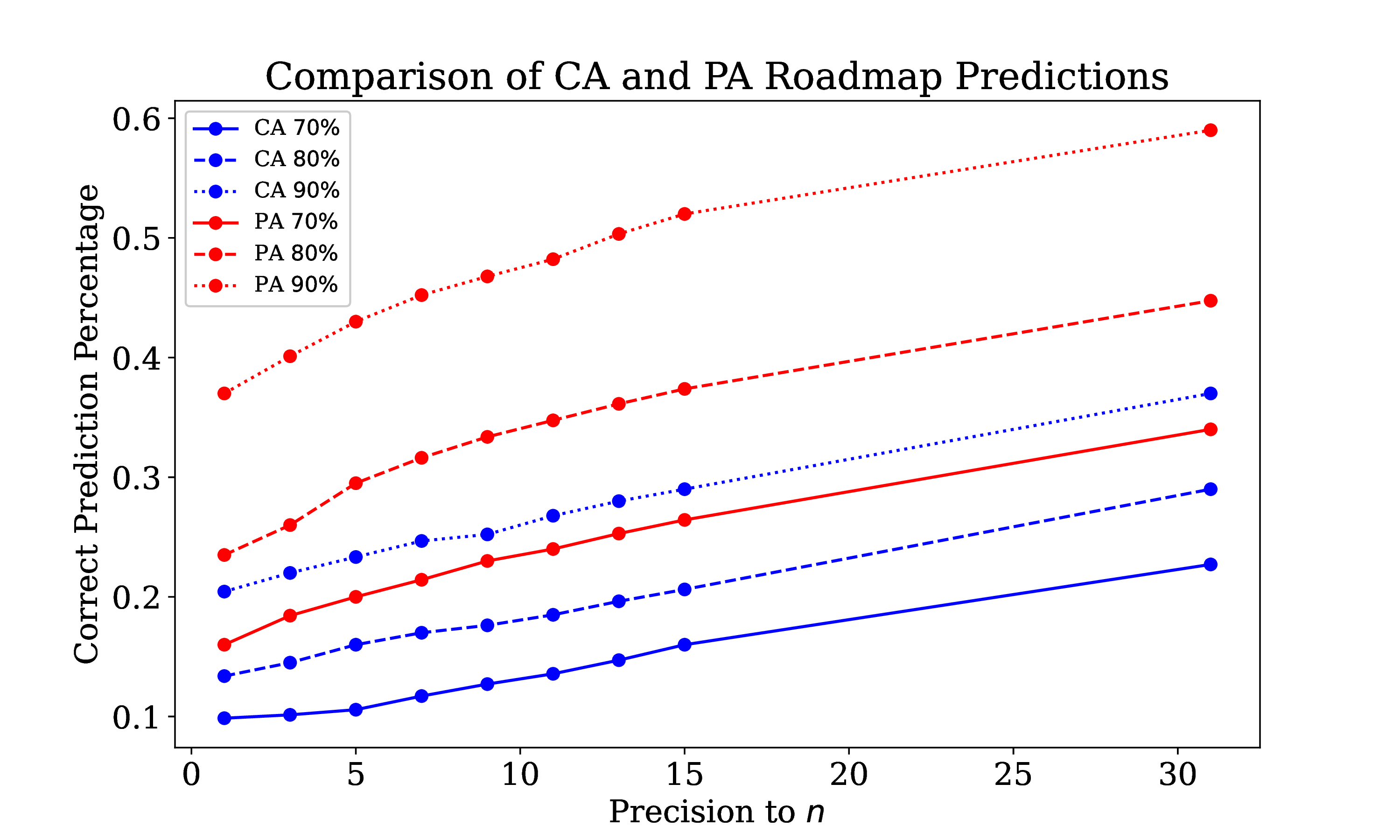}
        \caption{Roadmaps}
        \label{fig:subfigRoad}
        \vspace{0.6cm}
\end{figure}

Moreover, we show that DMC works better on graphs with high densities and clusterings in Table~\ref{tab:bioExperiments}. We provide results of applying DMC to more biological networks randomly chosen from the Brain Network and Biological Network sections in the Network Repository database. In Table~\ref{tab:bioExperiments}, ``grid-fission-yeast'', ``ce-gn'',``grid-human'', and ``yeast'' are from the Biological category, and ``mouse-retina-1'', ``fly-drosophila-medulla-1'', and ``mouse-kasthuri'' are from the Brain category.
\begin{table}
  \centering
  \caption{Table of experiments on biological networks, ordered by DMC's performance on $p = 90\%$ overlap subgraphs.}
  \begin{tabular}{lrrr}
    \toprule
    Networks & Score & Density & Clustering\\ \midrule
    grid-fission-yeast & .9932 &\textbf{.0123} & \textbf{.1874}\\
mouse-retina-1& .9900&\textbf{.9983} & \textbf{1.5539}\\
ce-gn& .9894&\textbf{.0218} & \textbf{.1839}\\
fly-drosophila-medulla-1& .9378&\textbf{.0211} & \textbf{3.8729}\\
grid-human& .8440&.0014 & .0800\\
mouse-kasthuri & .6156&.0032 & 0 \\ 
yeast & .5822&.0018 & .0708 \\ 
\bottomrule 
\end{tabular}
\label{tab:bioExperiments}
\end{table} 
We note that the highlighted densities are those above the 0.01 threshold, and for the average clustering coefficient the threshold is 0.1 (density and clustering taken from Network Repository). The graphs with both a high density and a high clustering coefficient outperformed the other graphs significantly. 

In this paragraph, we examine the behavior of DMC and Greedy DMC applied to the unweighted multi-layer PPI network under different conditions. For DMC, we see in Figure~\ref{fig:perfDecay} that DMC's performance decays as deletion probability (noise level) increases from 0.01 to 0.1. This is expected and no surprise for a new algorithm, but it is worth pointing out that the performance does not rapidly drop to near zero values even with ten times the experimental noise level, implying that the method has potential to withstand noise. In Table~\ref{tab:GreedyDMC}, we see that when $\epsilon = 0,10$, the Greedy DMC algorithm maintains high performance. We will study these special cases in the future. 
\begin{figure}
    \centering
\includegraphics[width=.35\textwidth]{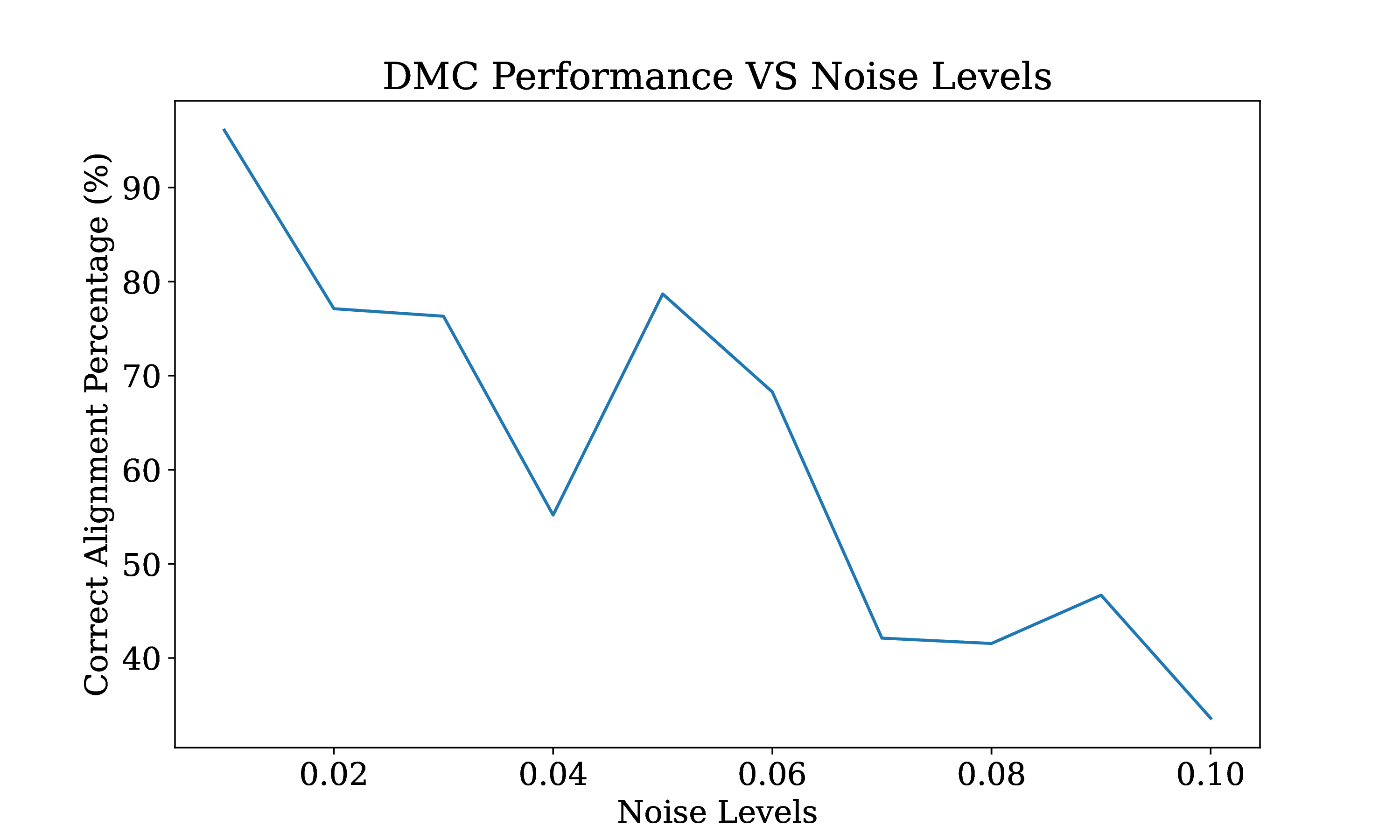}
    \caption{Performance Decay Plot}
    \label{fig:perfDecay}
    \vspace{0.6cm}
\end{figure}
\begin{table}
    \centering
    \caption{Greedy DMC Performance}
    \begin{tabular}{lr}
        \toprule
        $\epsilon$  & Score (\%) \\
        \midrule
        0  & 96.58\%\\
        10 & 96.32\%\\
        100 & 50.46\%\\
        1000 & 0.18\%\\
        \bottomrule
    \end{tabular}
    
    \label{tab:GreedyDMC}
\end{table}
\subsubsection{Artificial Networks}
In this section, we introduce the Erdös-Rényi, Barabási-Albert, and Chung-Lu models that were derived to generate complex networks, potentially to model real world network behaviors.  

The Erdös-Rényi model is the most basic random graph generation model. One common way to define this is to fix the number of vertices, say $n_{er}$, then connect an edge between every pair of vertices with some probability $p_{er}$. The resulting graph is typically denoted as $G(n_{er}, p_{er})$ \citep{Durrett2009}. One interesting property of these graphs is that their degree distributions are binomial, as can be be seen by the following calculation:
\begin{equation}
\textbf{P}(\text{deg}(v) = k) = \binom{n-1}{k}p^{k}(1-p)^{n-1-k}.
\label{eqtn:binomial}
\end{equation}
The probability considers the possibility of having degree of $k$ (being connected to $k$ edges) when the largest possible degree is $n-1$, given the total number of nodes is $n$. $p$ in (\ref{eqtn:binomial}) is the same $p_{er}$ for generating the random graph.

The Barabási-Albert model considers properties observed on real world networks, especially the preferential attachment phenomena \citep{Barabasi2016}. Instead of attaching new edges with uniform probability, it is possible to attach new edges with weighted probabilities. Suppose we already have $m$ nodes with some connections, then for the $(m+1)^{\text{th}}$ node, the probability that it will be attached to a previous node $v_{i}$ is 
\begin{equation}
\textbf{P}(v_{i}) = \dfrac{\text{deg}(v_{i})}{\Sigma_{i=1}^{m}\text{deg}(v_{i})}
\end{equation}
We can observe that the probabilities of attaching the new node to any previous node sums up to one, forming a probability distribution; moreover, higher degree nodes attract more new nodes. This usually leads to what we call \emph{hubs}, which are local centers that are attached to many low degree nodes. It is common for hubs to exist in biological and social networks, and hubs add to the heterogeneity of graphs. However, having hubs with too much concentration could lead to confusion when differentiating low-degree nodes. Therefore, it is always good to strike a balance between concentrating and diffusing nodes, but how this is done in the real world is yet unknown.

Now we introduce the Chung-Lu model. It also attempts to capture real world network properties, but instead of focusing on preferential attachment or community behaviors, they attempt to directly approximate real world networks' degree distributions \citep{Clauset2019}. Suppose we have a specific type of degree distribution in mind and we expect $m$ nodes in total, then the attachment probability between each node $v_{i}$ and another node $v_{j}$ is as follows:
\begin{equation}
\textbf{P}(ij) = \dfrac{\text{deg}(v_{i})\text{deg}(v_{j})}{\Sigma_{k=1}^{m}\text{deg}(v_{k})}
\label{eqtn:CLequation}
\end{equation}
In (\ref{eqtn:CLequation}), the degrees are target degrees according to our degree distribution in mind. With the attachment process described above, the expected degree at each node $v_{i}$ is
\begin{equation}
\Sigma_{k=1}^{m}\textbf{P}(ik) = \Sigma_{k=1}^{m}\dfrac{\text{deg}(v_{i})\text{deg}(v_{k})}{\Sigma_{k=1}^{m}\text{deg}(v_{k})} = \text{deg}(v_{i})
\end{equation}
which theoretically guarantees that the Chung-Lu model will generate a synthetic graph with a degree distribution that is close to our initial distribution in mind. 

Despite the careful construction of models, human knowledge of real world networks are nowhere near complete. Through comparison of DMC results on real world networks with experimental results on synthetic graphs, we will see that the complexity and irregularity of real world networks make their nodes more effectively differentiable through the DMC, and there is still a notable gap between the synthetic models and real world networks since they exhibit different behaviors in response to the same graph alignment method. 

Through examining existent literature, we summarize a set of (non-comprehensive) key words that are properties observed from real world networks: hierarchical \citep{Ravasz2003}, fractal \citep{Wen2021}, copying behavior \citep{Barabasi2016}, preferential attachment \citep{Barabasi2016}, small world \citep{Watts1998}, clustering \citep{Opsahl2009}, power-law degree distribution \citep{Barabasi2016}. These properties are of course not enough to capture the entirety of real world networks, which makes building new models to generate complex networks interesting.

\begin{definition}[k-partite Graph] A \emph{k-partite Graph} is a graph with nodes that can be split into $k$ independent sets, where nodes within each set are never connected. 
\end{definition}

We first applied DMC to graphs generated through the Erdös-Rényi and Barabási-Albert models. We examine the performance of DMC on the graphs with respect to different partite numbers $k$. For our Erdös-Rényi experiments, we fix total number of nodes to be $n_{er} = 100$, then split the nodes into $k$ partite groups; the partition of nodes is arbitrary. After that, we connect edges between different partite groups with probability of $p_{er} = 0.5$. Then we use the random walk graph sampling method to obtain two graphs with 90\% overlap, and perform DMC. We run the same experiment for 20 times, then take the mean performance. For our Barabási-Albert experiments, we do a similar thing. Every time we add a new node, we connect it to five different previous nodes (add 5 edges). 

We not only examined the completely random Erdös-Rényi and heterogeneous Barabási-Albert, but also experimented with intermediate graphs (i.e.~a spectrum of graphs ranging from most uniform to most heterogeneous). We move along the spectrum of Poisson distributions (see Figure~\ref{fig:Poisson}) with different parameter values, which can approximate distributions from binomial (black) to power law (yellow). A formal representation of the Poisson distribution for each node $v$ is as follows:
\begin{equation}
\text{P(deg($v$) = $k$)} = \dfrac{\lambda^{k}e^{-\lambda}}{k!}
\end{equation}
where $\lambda$ is a parameter. These pre-determined distributions are encoded in graphs through the Chung-Lu model. Figures~\ref{fig:CL_Partite_20}, \ref{fig:CL_Partite_10}, \ref{fig:CL_Partite_5}, and \ref{fig:CL_Partite_1} show how the behavior of DMC converges to the heterogeneous end.

Real world networks are more complex than we can imagine, hence the difficulty of capturing them with mathematical language. However, it is interesting the simple DMC works better on the more complex real world networks than synthetically generated graphs. This shows that modeling complex networks is no easy task, and understanding the complexity will potentially bring us rich information. Moreover, the fact that DMC has higher performance when using the more sophisticated Barabási-Albert and Chung-Lu, instead of Erdös-Rényi, indicates that it is no coincidence that DMC works well on heterogeneous real world networks. Additionally, while the Erdös-Rényi model was not as suitable for DMC, we discovered as a side result a linear relationship between the performance of DMC and the partite number of the Erdös-Rényi networks. 
\begin{figure}
    \centering
\includegraphics[width=.35\textwidth]{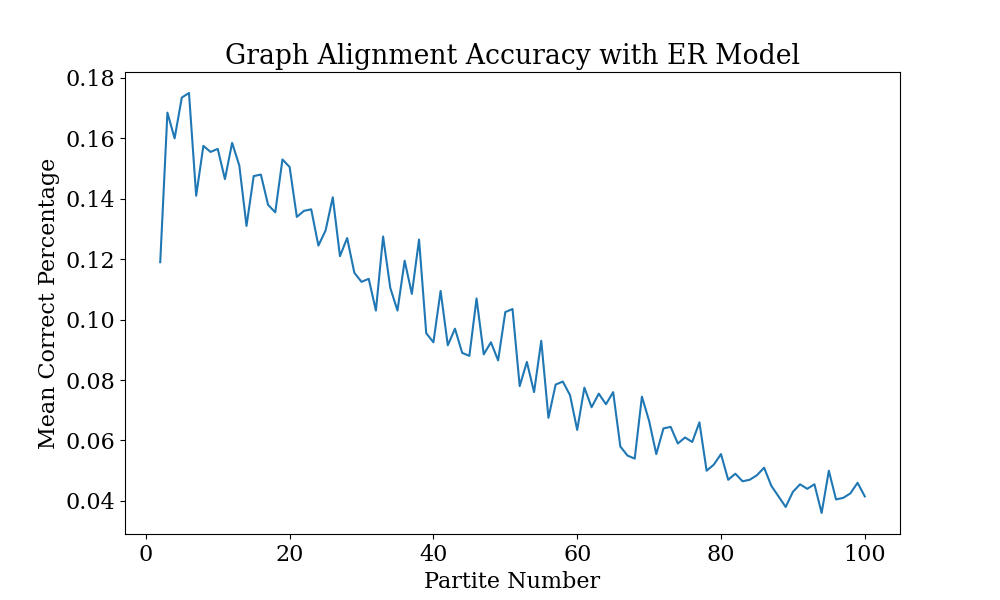}
    \caption{Erdos-Renyi Model: negative linear relationship.}
\end{figure}
\begin{figure}
    \centering
\includegraphics[width=.35\textwidth]{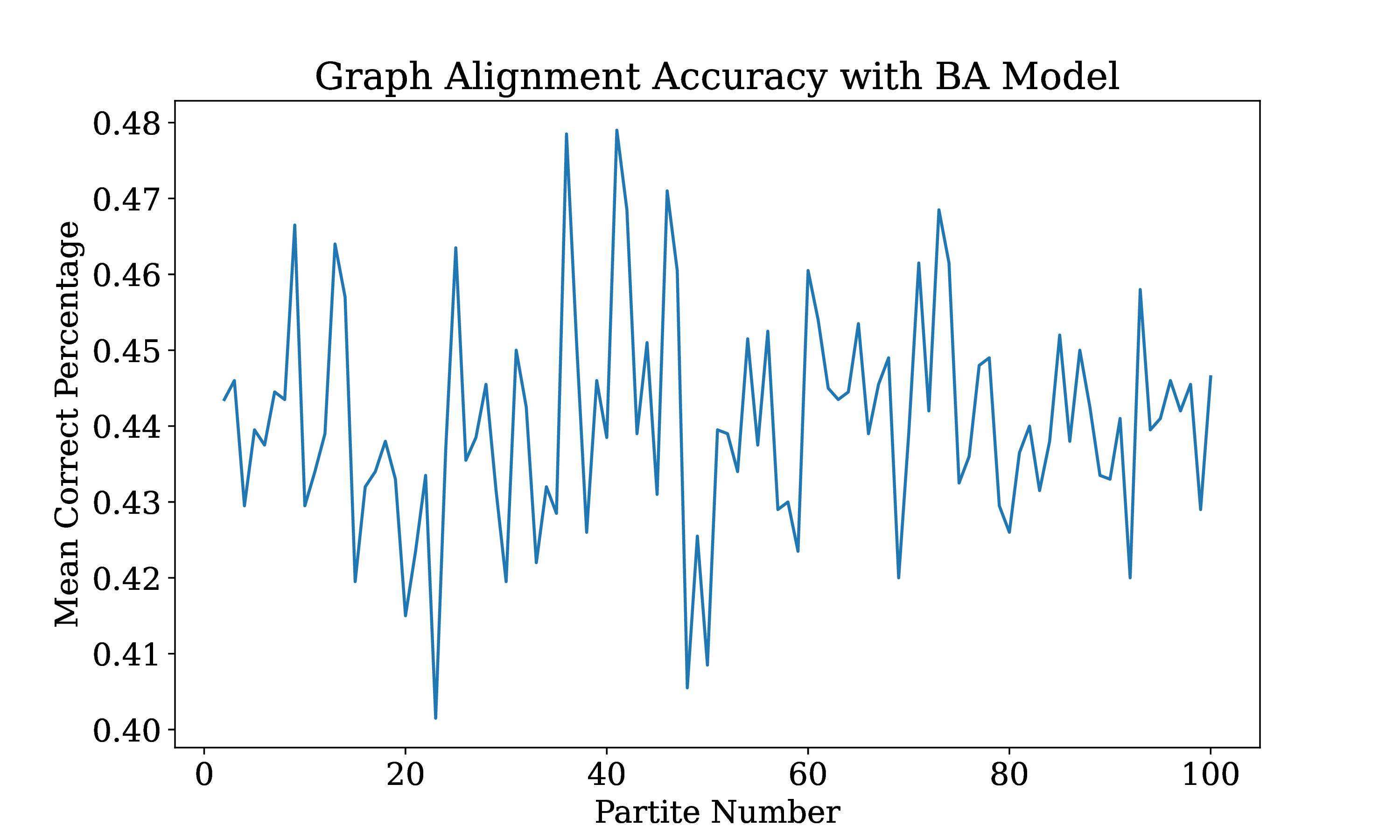}
    \caption{Barabasi-Albert Model}
\end{figure}

\begin{figure}[h!]
    \centering    \includegraphics[width=0.35\textwidth]{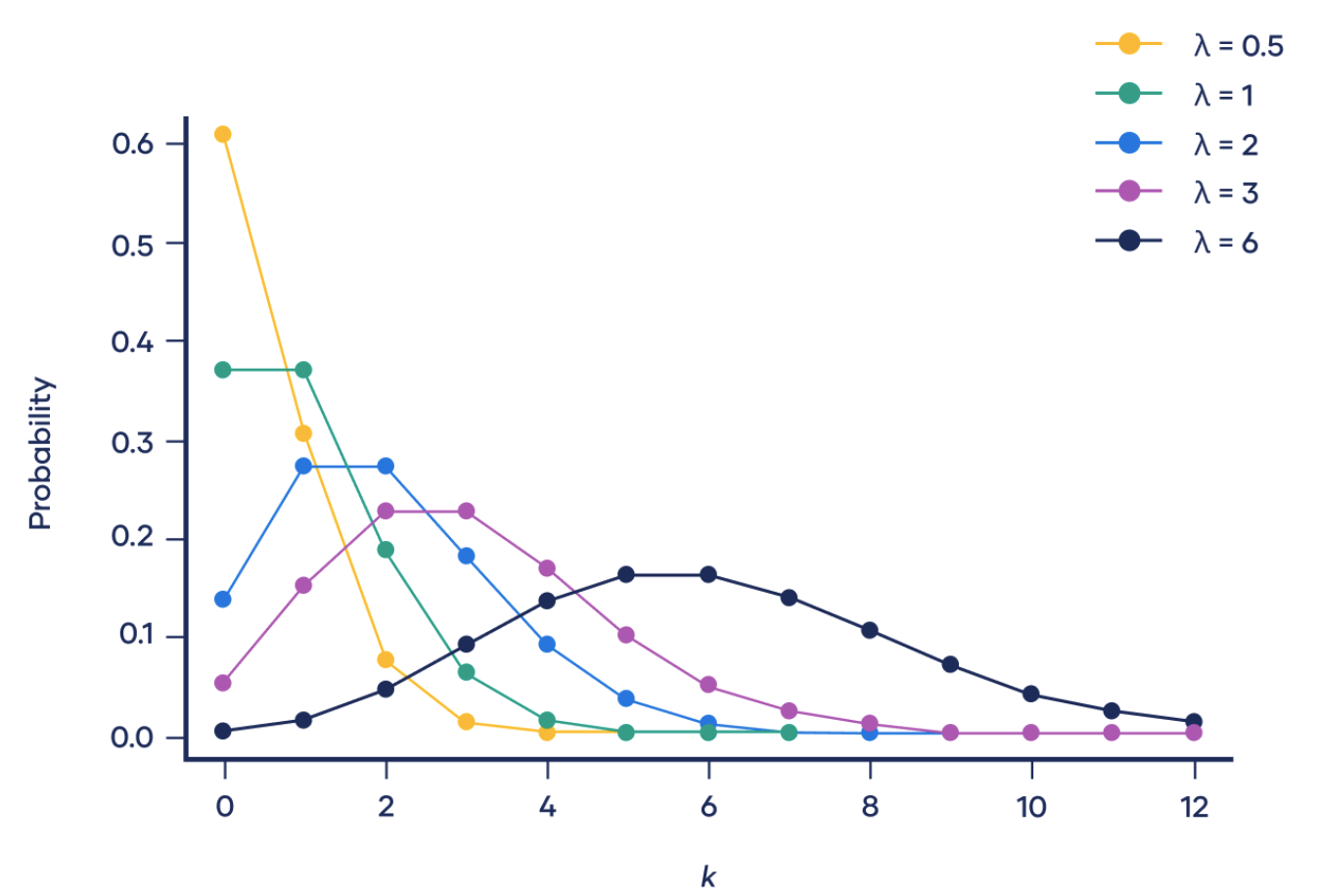}
    \caption{Evolution of Poisson distribution with decreasing $\lambda$. Taken from Scribbr.}
    \label{fig:Poisson}
\end{figure}
\begin{figure}
    \centering    \includegraphics[width=0.35\textwidth]{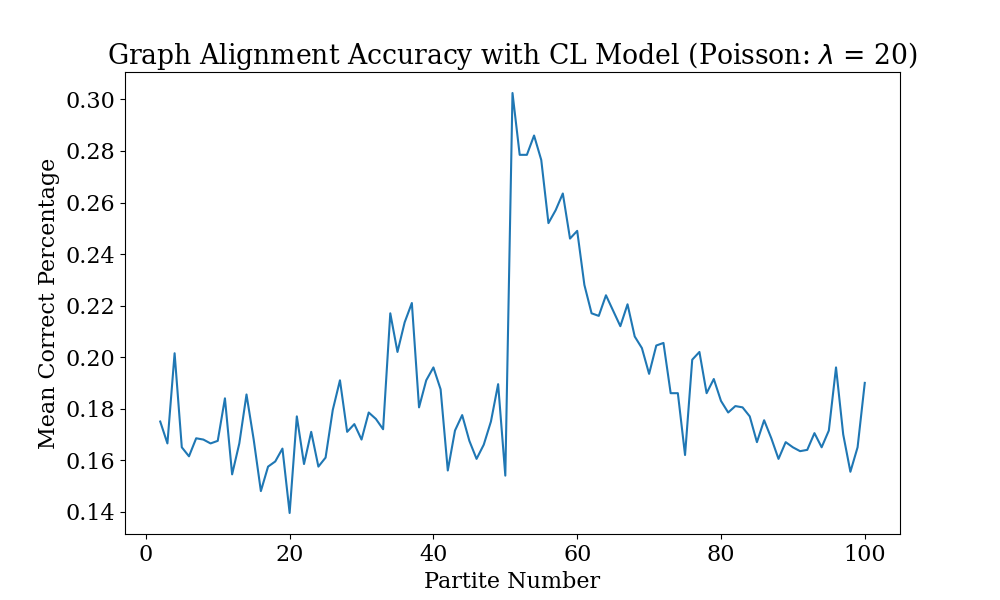}
    \caption{Performance of DMC applied to Chung-Lu generated graphs with respect to different partite numbers when parameter $\lambda=20$.}
    \label{fig:CL_Partite_20}
\end{figure}
\begin{figure}
    \centering    \includegraphics[width=0.35\textwidth]{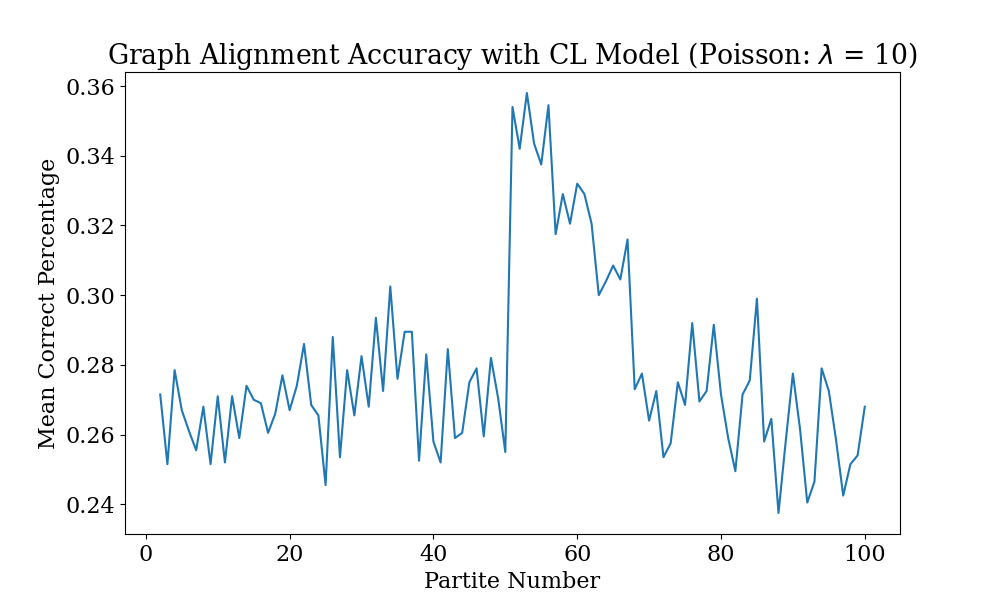}
    \caption{Performance of DMC applied to Chung-Lu generated graphs with respect to different partite numbers when parameter $\lambda=10$.}
    \label{fig:CL_Partite_10}
\end{figure}
\begin{figure}
    \centering    \includegraphics[width=0.35\textwidth]{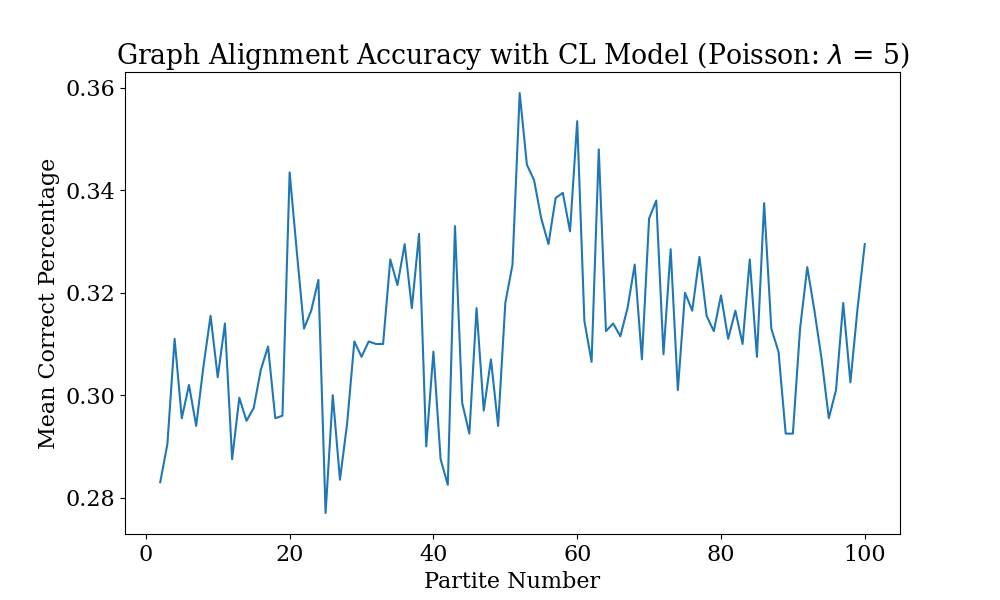}
    \caption{Performance of DMC applied to Chung-Lu generated graphs with respect to different partite numbers when parameter $\lambda=5$.}
    \label{fig:CL_Partite_5}
\end{figure}
\begin{figure}
    \centering    \includegraphics[width=0.35\textwidth]{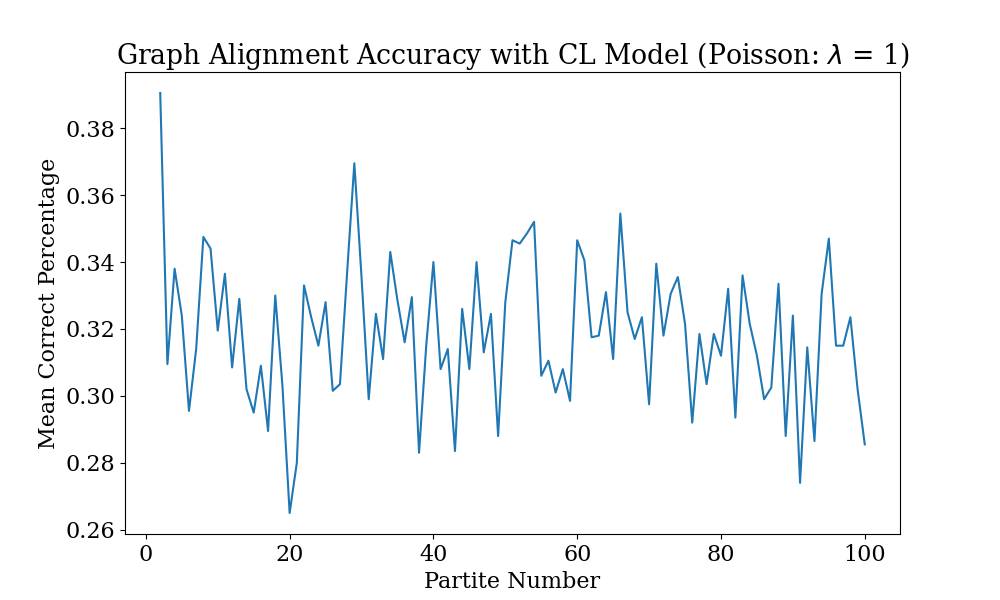}
    \caption{Performance of DMC applied to Chung-Lu generated graphs with respect to different partite numbers when parameter $\lambda=1$.}
    \label{fig:CL_Partite_1}
\end{figure}

\subsection{Weighted Graphs: Brief Examination of Weighted DMC}
As previously mentioned, we also dedicate a brief section to the Weighted DMC to demonstrate its potential, which in turn will also support that DMC, albeit its simplicity, is a valid approach. Table~\ref{tab:WeightedPPI0} provides the results of applying the Weighted DMC to five weighted PPI networks for different species that are typical subjects of biological studies, all taken from the STRING dataset \citep{STRING}. Among the five species, E.~Coli is the only prokaryotic organism, meaning that it has simpler cell structure lacking a sophisticated nucleus. We can reasonably speculate that the PPI for E.~Coli is less heterogeneous, leading to a less accurate Weighted DMC result. We have also discussed in Section 5.1 that heterogeneous networks tend to be more suitable for the simple DMC. The other four PPI networks all yielded performance above $95\%$, which is high.

We run the entire process, from random walk to graph alignment, for ten times and take their mean performance. All experiments sampled 5000-node graphs except for E.~Coli, where we sampled 3000-node graphs, due to its limited size. We used $p_{d}=0.01$ deletion probability for all experiments.
\begin{table}
    \centering
    \caption{Weighted PPI Networks}
    \begin{tabular}{lr}
        \toprule
        Species  & Score (\%) \\
        \midrule
        Human  & 99.19\\
        Yeast & 98.06\\
        Mouse & 97.27\\
        Fruit Fly & 96.26\\
        E. Coli & 91.46\\
        \bottomrule
    \end{tabular}
    \label{tab:WeightedPPI0}
\end{table}

We have also run experiments on social networks, where the results are reported in Table 6. ``soc-sign-bitcoin-otc'' and ``soc-sign-bitcoin-alpha'' are taken from the SNAP database \citep{SNAP}, and ``music-cotagging'', ``ai-cotagging'', ``apple-cotagging'', ``travel-cotagging'', and ``economics-cotagging'' are taken from the Computer Science Department webpage of Cornell University \citep{Fu2020}. The ``cotagging'' networks are networks by keywords on Stack Exchange. We used the edge deletion sampling method, setting deletion probability $p_{d} = 0.01$, and run all experiments ten times to take the mean performance. We can speculate here that including weight information of edges has helped us gain more insight into local graph structure, since the overall performance of the Weighted DMC on the social networks is higher than that of DMC applied to unweighted social networks. 
\begin{table}
    \centering
    \caption{More Weighted Networks}
    \begin{tabular}{lrr}
        \toprule
        Network  & Score (\%) & Sampled Nodes \\
        \midrule
        soc-sign-bitcoin-otc  & 77.44 & 5000\\
        soc-sign-bitcoin-alpha & 82.08 & 3000\\
        music-cotagging & 98.91 &467\\
        ai-cotagging & 96.95 & 285\\
        apple-cotagging & 98.65 & 1063\\
        travel-cotagging & 97.35 & 1779\\
        economics-cotagging & 98.27 & 369 \\
        \bottomrule
    \end{tabular}
    \label{tab:WeightedPPI}
\end{table}

\section{Conclusion and Future Work}
We proposed two graph alignment methods for unweighted networks: DMC and Greedy DMC, and also another one for weighted networks, which we named Weighted DMC. We showed that they were theoretically motivated, and analyzed experimental results. DMC's accuracy on classic PPI networks proves to be higher than carefully chosen baselines, and it works best for heterogeneous graphs which has been demonstrated from multiple perspectives. This process included providing an overview of synthetic graph generation models and a short discussion of the difficulty of modeling complex networks. Future works could examine special cases of Greedy DMC and the Weighted DMC, as pointed out in previous sections; one could also test this sequence of methods with more datasets to more clearly define the optimal set of networks that should be aligned using DMC, a simple method that can easily be replicated. 


 






\bibliography{mybibfile}

\end{document}